\documentclass[letterpaper,10pt,conference]{ieeeconf}  

\IEEEoverridecommandlockouts                              

\usepackage{graphics} 
\usepackage{epsfig} 
\usepackage{amsmath} 
\usepackage{amssymb}  
 
\usepackage{amsthm}
\usepackage{cite}
\usepackage{color}
\usepackage{enumerate}
\usepackage{graphicx}
\usepackage{mathtools}
\usepackage{cancel}
\usepackage{algorithm}
\usepackage{algpseudocode}
\usepackage{textgreek}
\usepackage{setspace}
\usepackage{cleveref}
\crefname{equation}{}{}
\crefname{proposition}{Proposition}{Propositions}
\usepackage{url}
\newcommand{\nspace}[1]{\mkern#1mu}
\graphicspath{{figures/}}

\newcommand{\real}{\mathbb{R}}
\newcommand{\realnonneg}{\mathbb{R}_{\ge 0}}

\newcommand{\until}[1]{[#1]}
\newcommand{\map}[3]{#1:#2 \rightarrow #3}

\newcommand{\longthmtitle}[1]{\mbox{}{\textit{(#1):}}}

\newcommand{\setdefb}[2]{\big\{#1 \; | \; #2\big\}}
\newcommand{\setdefB}[2]{\Big\{#1 \; | \; #2\Big\}}
\newcommand*{\SetSuchThat}[1][]{} 
\newcommand*{\MvertSets}{%
    \renewcommand*\SetSuchThat[1][]{%
        \mathclose{}%
        \nonscript\;##1\vert\penalty\relpenalty\nonscript\;%
        \mathopen{}%
    }%
}
\MvertSets 
\DeclarePairedDelimiterX \Set [2] {\lbrace}{\rbrace}
    {\,#1\SetSuchThat[\delimsize]#2\,}

\newcommand{\Ac}{\mathcal{A}}
\newcommand{\Cc}{\mathcal{C}}

\newcommand{\Ec}{\mathcal{E}}
\newcommand{\Fc}{\mathcal{F}}

\newcommand{\Ic}{\mathcal{I}}
\newcommand{\Kc}{\mathcal{K}}

\newcommand{\Uc}{\mathcal{U}}
\newcommand{\Sc}{\mathcal{S}}

\newcommand{\R}{\mathbb{R}}

\newcommand{\Ke}{\Kc^{\rm e}}

\newcommand{\defeq}{\triangleq}

\newtheorem{theorem}{Theorem}
\newtheorem{lemma}{Lemma}
\newtheorem{corollary}{Corollary}
\newtheorem{proposition}{Proposition}

\theoremstyle{definition}
\newtheorem{definition}{Definition}

\newtheorem{remark}{Remark}

\newtheorem{assumption}{Assumption}

\newcommand{\be}{\mathbf{e}}
\renewcommand{\bf}{\mathbf{f}} 
\newcommand{\bg}{\mathbf{g}}

\newcommand{\bk}{\mathbf{k}}

\newcommand{\bu}{\mathbf{u}}
\newcommand{\bv}{\mathbf{v}}

\newcommand{\bx}{\mathbf{x}}
\newcommand{\by}{\mathbf{y}}

\newcommand{\bA}{\mathbf{A}}
\newcommand{\bB}{\mathbf{B}}

\newcommand{\bI}{\mathbf{I}}
\newcommand{\bJ}{\mathbf{J}}

\newcommand{\bP}{\mathbf{P}}
\newcommand{\bQ}{\mathbf{Q}}
\newcommand{\bR}{\mathbf{R}}

\newcommand{\bOmega}{\boldsymbol{\Omega}}
\newcommand{\bphi}{\boldsymbol{\varphi}}

\newcommand{\obs}{{\operatorname{o}}}
\newcommand{\des}{{\operatorname{d}}}

\newcommand{\diag}{{\operatorname{diag}}}

\renewcommand{\baselinestretch}{1}
\allowdisplaybreaks

\DeclareMathOperator*{\argmin}{argmin}

\title{\LARGE \textbf{SafeSpace: Aggregating Safe Sets from Backup \\ Control Barrier Functions under Input Constraints}}

\author{Pio Ong, David E. J. van Wijk, Massimiliano de Sa, Joel W. Burdick, Aaron D. Ames %
\thanks{All authors are with the Department of Mechanical and Civil Engineering, California Institute of Technology, Pasadena, CA 91125, USA. \texttt{\{pioong, vanwijk, mdesa, jburdick, ames\}@caltech.edu}.}
\thanks{This work was in part supported by the Technology Innovation Institute, AFOSR Award \#113535-19668, and DARPA under the LINC program.}
}
\begin{document}
\maketitle
\begin{abstract}
Control barrier functions (CBFs) provide a principled framework for enforcing safety in control systems---yet the certified safe operating region in practice is often conservative, especially under input bounds. In many applications, multiple smaller safe sets can be certified independently, e.g., around distinct equilibria with different stabilizing controllers. This paper proposes a framework for uniting such regions into a single certified safe set using \emph{combinatorial CBFs}. We refine the combinatorial CBF framework by introducing an auxiliary variable that enables logical compositions of individual CBFs. In the proposed framework, we show that such compositions yield a \emph{generalized combinatorial CBF} under a condition termed \emph{conjunctive compatibility}. Building on this result, we extend the framework to enable the aggregation of multiple implicit safe sets generated by the backup CBF framework. We show that the resulting CBF-based quadratic program yields a continuous safety filter over the aggregated safe region. The approach is demonstrated on two spacecraft safety problems, safe attitude control and safe station keeping, where multiple certified safe regions are combined to expand the operational envelope.
\end{abstract}

\begin{spacing}{0.975}
\section{Introduction}
Control barrier functions (CBFs)~\cite{ADA-XX-JWG-PT:17} have emerged as a powerful tool for enforcing safety in control systems by enabling the synthesis of controllers that render a desired safe set forward invariant. Despite their effectiveness, the certified safe operating region obtained in practice is often conservative, particularly in the presence of input bounds. This typically leads to the identification of multiple smaller safe sets, each associated with a distinct controller or safe operating condition. However, safety filters constructed for a given set confine the system to its initial region, limiting the flexibility to transition between multiple certified regions. This motivates the development of systematic methods for enlarging certified safe regions, either by expanding a given safe set or by combining multiple safe sets into a single one.

A prominent class of approaches to enlarging safe operating regions seeks to approximate the viability kernel~\cite{JPA-AMB-PSP:11}, i.e., the maximal control invariant set contained within the safety constraint. Methods based on Hamilton--Jacobi reachability \cite{SB-MC-SH-CJT:17} aim to characterize such sets through propagating backward reachability of a given initial set under all admissible control inputs. While these methods converge to the desired viability kernel, they require solving partial differential equations, leading to significant computational complexity and limited scalability. On a related note, backup CBFs \cite{TG-MM-AS-PN-EF-ADA:20} offer a scalable alternative by expanding  a safe set using a \textit{single} controller that renders it forward invariant, resulting in tractable characterization of the expanded safe set. However this restriction to a single controller generally leads to conservative approximations of the viability kernel. To mitigate this, recent work~\cite{PR-JBH:25} has explored combining multiple backup controllers and their associated safe sets. Such development highlights the potential of aggregating multiple certified regions.

An alternative approach is to combine multiple safety constraints through logical compositions of control barrier functions. Early works~\cite{PG-JC-ME:17, LW-ADA-ME:16} formulate such compositions using pointwise maximum operators to encode disjunctions, corresponding to the union of safe sets. However, these constructions introduce nonsmoothness in the resulting barrier functions, which complicates both analysis and controller synthesis. Smooth approximations based on soft maximum operators~\cite{TGM-ADA:23, TGM:25, MB-DP:23} mitigate this issue at the cost of additional conservatism, and have also been used to combine safe sets from backup controllers~\cite{PR-JBH:25}. Related formulations based on signal temporal logic~\cite{LL-DVD:19,LL-DVD:19b} provides expressive tools for combining safety requirements, but similarly rely on approximations or nonsmooth constructions. Our work builds on the more recent development of combinatorial CBFs~\cite{PO-HL-TGM-DP-ADA:25-csl}, which enable logical compositions without these limitations.

The logical combination of multiple safety constraints also raises feasibility challenges, as the corresponding conditions must be satisfied simultaneously. In standard CBF-CLF quadratic programs that incorporate both stability and safety constraints~\cite{ADA-SC-ME-GN-KS-PT:19}, feasibility is often ensured by introducing slack variables, sacrificing satisfaction of some constraints. More recently, the optimal-decay control barrier function (OD-CBF) framework~\cite{JZ-BZ-ZL-KS:21} improved feasibility for individual safety constraints by treating the decay rate as a decision variable. This additional flexibility in the decay rate helps bridge the gap between control invariance and the standard barrier condition; see the complete characterization provided in~\cite{PO-MHC-TGM-ADA:25-cdc} or the discussion on converse theorems in~\cite{PM-JC:25}. In this work, we extend this idea by introducing auxiliary variables to enable compatibility among multiple constraints and ensure the feasibility of the resulting safety filter.

\textbf{Statement of Contribution:} This paper proposes a unified framework for enlarging certified safe regions by combining multiple safety constraints within the control barrier function paradigm. First, we refine the combinatorial CBF framework by introducing an auxiliary variable that enables logical compositions of multiple CBFs. The resulting \emph{generalized combinatorial CBF} framework ensures that multiple constraints can be simultaneously enforced under a condition termed \emph{conjunctive compatibility}. Second, we extend the framework for backup CBFs, enabling logical combinations of implicit safe sets. Third, we show that the associated CBF-based quadratic program is a continuous feedback law over the aggregated safe set, avoiding the need for switching or hybrid control strategies. Finally, we demonstrate the proposed approach on spacecraft safety problems, including safe attitude control and safe station keeping, where only small certified safe regions can be found due to severe input bounds. Simulation results show that combining multiple safe sets significantly expands the safe operational region when compared to the standard CBF method.

\section{Background}

\subsection{Safety via Control Barrier Functions}\label{sec:vanillaCBF}

Consider the control-affine system\footnote{We denote $[N] \defeq \{1,\dots,N\}$. Given a differentiable function $\map{h}{\real^n}{\real}$, its Lie derivatives along a vector field $\map{\bf}{\real^n}{\real^n}$ (or $\map{\bf}{\real^n}{\real^{n\times m}}$) is defined as $L_\bf h(\bx) \defeq \nabla h(\bx) \bf(\bx)$.
A function $\alpha:\real\to\real$ is of class-$\Ke$ if it is continuous, strictly increasing, and satisfies $\alpha(0)=0$.
Given a finite set $\Fc\subset\real$, $\max^r \Fc$ denotes its $r$-th order statistic, i.e., the $r$-th largest element in $\Fc$.}:
\begin{equation}\label{sys:ctrl_affine}
    \dot\bx = \bf(\bx)+\bg(\bx)\bu,
\end{equation}
with system state $\bx\in\real^n$ and control input $\bu\in\Uc$, where $\Uc \subseteq\real^m$ is a polytopic set. The drift dynamics $\bf:\real^n \to \real^n$ and the control matrix $\bg:\real^n \to \real^{n \times m}$ are assumed to be continuously differentiable. For this system, we consider state constraints defined by a \textit{safety function} $\map{\psi}{\real^n}{\real}$:
\begin{equation}\label{eq:safety_constraint}
    \Sc \defeq \setdefb{\bx\in\real^n}{\psi(\bx)\geq 0}.
\end{equation}
We seek to design a controller $\map{\bk}{\real^n}{\Uc}$ such that the state feedback $\bu=\bk(\bx)$ produces closed-loop  trajectories $\bphi_\bk(t,\bx_0)$ that remain inside the set $\Sc$ for all $t\geq 0$ and initial conditions $\bx_0\in\Cc$. A necessary condition for the existence of such a controller is that the set $\Sc$ is control invariant.
\begin{definition}\longthmtitle{Control Invariance}\label{def:contr-invariance}
    A set $\Cc$ is \textit{control invariant} for system~\eqref{sys:ctrl_affine} under a given input constraint $\Uc\subseteq \real^m$ if, for each initial condition $\bx_0\in\Cc$, there exists a control signal $t\mapsto\bu(t)$ satisfying $\bu(t)\in\Uc$ for all $t\geq 0$ such that the corresponding state trajectory $\bphi_\bu(t,\bx_0)$ remains inside the set $\Cc$ for all $t\geq 0$.~\hfill$\diamond$
\end{definition}
In many applications, the constraint set~$\Sc$ is specified by the problem requirements and is generally not control invariant by default. To address safety concerns, we therefore seek a control invariant $\Cc$:
\begin{equation}\label{eq:safe_set}
    \Cc \defeq \setdefb{\bx\in\real^n}{h(\bx)\geq 0}.
\end{equation}
defined by a function $\map{h}{\real^n}{\real}$. The set $\Cc$ is said to be \textit{safe} if it is control invariant and satisfies $\Cc\subseteq\Sc$. This set can be interpreted as a safe operating region for the system.

We next introduce control barrier functions, which provide sufficient conditions to certify the control invariance of $\Cc$.
\begin{definition}\longthmtitle{Control Barrier Functions~\cite{ADA-XX-JWG-PT:17}}
    A continuously differentiable function $\map{h}{\real^n}{\real}$ is a \textit{control barrier function (CBF)} for system~\eqref{sys:ctrl_affine} under a given input constraint $\Uc\subseteq \real^m$ if there exists a class-$\Ke$ function $\alpha$ such that, for each $\bx \in \Cc$, there exists a control input $\bu\in\Uc$ satisfying:
    \begin{equation}
    \label{eq:CBF}
        \dot h(\bx,\bu)\defeq L_\bf h(\bx)+L_\bg h(\bx)\bu > -\alpha(h(\bx)).
    \end{equation}
\end{definition}
The key idea underlying CBFs is that they ensure the existence of a control input $\bu\in\Uc$ such that the differential inequality ``${\dot h > -\alpha(h)}$" is satisfied, which, by the comparison lemma, guarantees that $h(\bx(t))$ remains nonnegative for all~${t\geq 0}$. Note, however, that the CBF condition itself only guarantees the existence of an admissible control input at each state ${\bx \in \Cc}$; through an optimization-based controller synthesis framework, these pointwise inputs can be assembled into a continuous state-feedback function.  

\subsection{Safety Filter Framework}\label{sec:filter}
A common approach to enforcing safety constraints using CBFs is through the \emph{safety filter} framework \cite{ADA-SC-ME-GN-KS-PT:19}. Given a continuous nominal controller ${\bk_{\des}:\real^n \to \real^m}$, a widely used construction of a safety-filtered controller ${\bx\mapsto\bk(\bx)}$ defined on a neighborhood of $\Cc$, respecting a polytopic input constraint~$\Uc$, is given by the quadratic program (CBF-QP):
\begin{align}\label{eq:CBF-QP}
    \bk(\bx)  =\argmin_{\bu\in\Uc} \quad & \|\bu-\bk_\des(\bx)\|^2\\
     \textup{s.t.} \quad & \dot h_j(\bx,\bu) \geq -\alpha(h_j(\bx)),~\forall j\in\until{p} \nonumber
\end{align}
which naturally accommodates multiple CBFs $\{h_j\}_{j=1}^p$ defining sets $\{\Cc_j\}_{j=1}^p$ as in~\eqref{eq:safe_set}. The resulting controller is continuous provided that the CBF-QP satisfies Slater's condition, i.e., there exists a control input $\bu$ that strictly satisfies all inequalities\footnote{This observation motivates the use of a strict inequality in the modern definition of CBF~\eqref{eq:CBF}.}, at each $\bx\in\Cc$~\cite{PM-AA-JC:25}. Moreover, by construction, the CBF-QP controller~\eqref{eq:CBF-QP} satisfies all the CBF constraints and therefore renders the set 
\begin{align}
\bigcap_{j\in[p]}\Cc_j &= \setdefb{\bx\in\real^n}{h_j(\bx)\geq 0,~\forall j\in\until{p}}\nonumber , \\
&= \setdefb{\bx\in\real^n}{\min_{j\in\until{p}}h_j(\bx)\geq 0},
\end{align}
forward invariant. With the proposed cost function, the safety filter achieves this while minimally modifying the nominal control input $\bk_\des(\bx)$ at each state~$\bx \in \Cc$.

The formulation above enforces safety for a conjunctive combination of individual safety constraints by ensuring that all CBF $h_j$ remain nonnegative along the trajectories. Beyond conjunctions, the combinatorial CBF framework~\cite{PO-HL-TGM-DP-ADA:25-csl} generalizes the standard CBF condition~\eqref{eq:CBF} by introducing an additional term, enabling more general logical combinations of safe sets. In particular, the resulting CBF-QP is given by:
\begin{align}\label{eq:comb-CBF-QP}
    \bk(\bx)  =\argmin_{\bu\in\Uc} \quad & \|\bu-\bk_\des(\bx)\|^2\\
     \textup{s.t.} \quad  \dot h_j(\bx,\bu) &\geq -\alpha(h(\bx)+|h_j(\bx)-h(\bx)|),~\forall i\in\until{p} \nonumber
\end{align}
where $h$ defining the safe set $\Cc$ as in~\eqref{eq:safe_set} is constructed via sorting \textit{primitive} CBFs $h_j$ to represent different logical combinations. For example, $h(\bx) = \min_{i\in\until{p}}h_i(\bx)$ corresponds to conjunction as seen earlier, while $h(\bx) = \max_{i\in\until{p}}h_i(\bx)$ corresponds to disjunction.

Despite the sophistication of existing CBF-based frameworks, the certified safe operating region~$\Cc$ obtained in practice is often conservative. For instance, a common approach to constructing such a set is to identify a stabilizing controller for a safe equilibrium point and to estimate its region of attraction (e.g., by using Lyapunov sublevel sets, see~\cite{LG-AKK-TGM:25}) contained within the constraint set~$\Sc$. Such estimates are conservative, and the resulting region is further reduced when input constraints are taken into account.

\subsection{Set Expansion via Backup CBF Framework}\label{sec:vanillaBackup}
The backup CBF framework~\cite{TG-MM-AS-PN-EF-ADA:20} aims to reduce this conservatism by leveraging the known safe region and an associated safeguarding controller. In particular, suppose that a set $\Cc$ in~\eqref{eq:safe_set} can be rendered forward invariant with a known \textit{backup} controller $\map{\bk_{\rm b}}{\real^n}{\Uc}$ through the state-feedback $\bu=\bk_{\rm b}(\bx)$. Consider the resulting \textit{backup} system dynamics:
\begin{equation}\label{sys:backup_dyn}
     \dot \bx = \bf_{\rm b}(\bx)\triangleq \bf(\bx)+\bg(\bx)\bk_{\rm b}(\bx).
\end{equation}
Let $\bphi_{\bk_{\rm b}}\!(t,\bx_0)$ be the \textit{backup} trajectory generated by \eqref{sys:backup_dyn} from an initial condition ${\bx_0\in\real^n}$. The set of states visited by this trajectory~${\Omega(\bx_0) \! \defeq \!\setdefb{\bphi_{\bk_{\rm b}}\!(\tau,\bx_0)}{ \tau \geq 0}}$ is control invariant. 

For tractability of subsequent results, we consider finite-horizon backup trajectories over a time interval $[0,T]$ for some $T>0$.
We define the \textit{implicit} safe set:
\begin{align}
    \label{eq:implicit_safe_set}
    \Cc_{\rm I} &\defeq \left\{\bx\in \real^n~\left|~
    \begin{array}{c}
    \psi(\bphi_{\bk_{\rm b}}\!(\tau,\bx))\geq 0,~\forall \tau\in[0,T] \\
    h(\bphi_{\bk_{\rm b}}\!(T,\bx))\geq 0 \\
    \end{array}
    \right.\right\} \nonumber, \\
    &= \setdefb{\bx\in \real^n}{h_{\rm I}(\bx)\geq 0},
\end{align}
where the function:
\begin{equation}\label{eq:implicit_CBF}
    h_{\rm I}(\bx) \defeq \min \Big\{\min_{\tau\in[0,T]}\psi(\bphi_{\bk_{\rm b}}\!(\tau,\bx)),h(\bphi_{\bk_{\rm b}}\!(T,\bx))\Big\},
\end{equation}
is introduced to facilitate the derivation of CBF conditions for controller synthesis. This set, ${\Cc_{\rm I} \subseteq \Sc}$, denotes the set of all states $\bx$ from which the control invariant set $\Cc$ can be reached \textit{safely}, under the controller $\bk_{\rm b}$. By expanding $\Cc$ with $\bk_{\rm b}$, $\Cc_{\rm I}$ is also control invariant \cite{TG-MM-AS-PN-EF-ADA:20}, and is rendered forward invariant by the backup controller. Since $h_{\rm I}$ is generally nonsmooth, it requires a set of conditions distinct from~\eqref{eq:CBF}.
\begin{definition}\longthmtitle{Implicit CBF}
    Given a continuously differentiable backup controller $\map{\bk_{\rm b}}{\real^n}{\Uc}$ that renders the set $\Cc$ in~\eqref{eq:safe_set} forward invariant, the function $\map{h_{\rm I}}{\real^n}{\real}$ in~\eqref{eq:implicit_CBF} is called an \textit{implicit CBF} for~\eqref{sys:ctrl_affine} associated with the controller $\bk_{\rm b}$ if there exists a class-$\Ke$ function $\alpha$ such that, for each $\bx\in\Cc_{\rm I}$, there exists a control input $\bu\in\Uc$ satisfying:
    \begin{subequations}\label{eq:implicit_CBF_cond}
    \begin{align}
        &\dot \psi_{\bk_{\rm b}}(\tau,\bx,\bu) > -\alpha(\psi(\bphi_{\bk_{\rm b}}(\tau, \bx))),~\forall \tau\in[0,T] \label{eq:traj_continuous}\\
        &\dot h_{\bk_{\rm b}}(T,\bx,\bu) > -\alpha(h(\bphi_{\bk_{\rm b}}(T, \bx))),
    \end{align}
    \end{subequations}
    where we use the shorthand notation:
    \begin{subequations}
        \begin{align}
            \dot \psi_{\bk_{\rm b}}(\tau,\bx,\bu)\defeq~ &L_\bf (\psi\circ \bphi_{\bk_{\rm b}}(\tau,\cdot))(\bx) \\
             &\quad+ L_\bg (\psi\circ \bphi_{\bk_{\rm b}}(\tau,\cdot))(\bx)\bu \nonumber,\\
            \dot h_{\bk_{\rm b}}(T,\bx,\bu)\defeq~
            &L_\bf (h\circ \bphi_{\bk_{\rm b}}(T,\cdot))( \bx) \\
            &\quad+ L_\bg (h\circ \bphi_{\bk_{\rm b}}(T,\cdot))(\bx)\bu \nonumber.
        \end{align} 
    \end{subequations}
\end{definition}
The function $h_{\rm I}$ is termed \emph{implicit} since the associated safe set $\Cc_{\rm I}$ is defined through the backup flow which can not typically be expressed in closed form. In principle, any controller satisfying the implicit CBF conditions~\eqref{eq:implicit_CBF_cond} guarantees forward invariance of the safe set $\Cc_{\rm I}$, cf.~\cite{TG-MM-AS-PN-EF-ADA:20}.
In practice, however, backup trajectories $\bphi_{\bk_{\rm b}}$ must be computed numerically and can only be evaluated at a finite number of points along a trajectory. To this end, we consider a uniform discretization of the interval $[0,T]$ with step size $\Delta \tau > 0$, and evaluate the trajectory at sampling times ${\tau_k = k\Delta \tau}$ for ${k \in \until{N}}$, where ${N \Delta \tau = T}$. The infinite collection of trajectory-level constraints in~\eqref{eq:implicit_CBF_cond} is approximated by a finite set of constraints enforced at these discrete sampling points:
\begin{subequations}
    \begin{align}
        &\dot \psi_{\bk_{\rm b}}(\tau_k,\bx,\bu) \geq -\alpha(\psi(\bphi_{\bk_{\rm b}}(\tau_k, \bx))),~\forall k\in\until{N} \label{eq:traj_discrete}\\
        &\dot h_{\bk_{\rm b}}(T,\bx,\bu) \geq -\alpha(h(\bphi_{\bk_{\rm b}}(T, \bx))).
    \end{align}
\end{subequations}
We accordingly define the \textit{discretized implicit safe set}:
\begin{align}
    \label{eq:implicit_safe_set_finite}
    \Cc_{{\rm I}}^{\rm fin} &\defeq \left\{\bx\in \real^n \left|
    \begin{array}{l}
    \psi(\bphi_{\bk_{\rm b}}\!(\tau_k,\bx))\geq 0,~\forall k\in \until N \\
    h(\bphi_{\bk_{\rm b}}\!(T,\bx))\geq 0 \\
    \end{array}
    \right.\right\} \nonumber,\\
    &= \setdefb{\bx\in \real^n}{h_{\rm I}^{\rm fin}(\bx)\geq 0},
\end{align}
with the function:
\begin{equation}\label{eq:implicit_CBF_finite}
    h_{\rm I}^{\rm fin}(\bx) \defeq \min \Big\{\min_{k \in \until N}\psi(\bphi_{\bk_{\rm b}}\!(\tau_k,\bx)),h(\bphi_{\bk_{\rm b}}\!(T,\bx))\Big\}.
\end{equation}
As safety conditions are considered only at discrete sampling times, the set $\Cc_{{\rm I}}^{\rm fin}$ is a relaxation of the continuous implicit safe set $\Cc_{\rm I}$. In this work, we adopt $\Cc_{{\rm I}}^{\rm fin}$ as the certified safe operating region and assume that rendering $\Cc_{{\rm I}}^{\rm fin}$ forward invariant is sufficient for the intended safety specification.

\begin{assumption}\longthmtitle{Practical Safety}\label{assump:implicit_finite_safe}
$\Cc_{{\rm I}}^{\rm fin}$ retains the control invariance property from $\mathcal{C}_{\rm I}$, and the satisfaction of ${\psi(\bphi_{\bk_{\rm b}}\!(\tau_k,\bx))\geq 0, \forall \nspace{1} k\in \until N}$ implies the satisfaction of $\psi(\bphi_{\bk_{\rm b}}\!(\tau,\bx))\geq 0,\forall \tau\in[0,T]$.
\end{assumption}
Recent work~\cite[Lem. 1]{TG-MM-AS-PN-EF-ADA:20} has investigated  formulations that account for the discretization by robustifying each constraint with additional margin terms. In this paper, we make the above simplifying assumptions and focus on the problem of handling multiple backup controllers for safe set expansion.

\section{Uniting Multiple Safe Sets}\label{sec:gen-combo-cbf}

Before addressing the implicit case, we first study how to unite multiple explicitly defined safe sets. Suppose that $p$ safe sets $\{\Cc_j\}_{j=1}^p$ are given, each of which is control invariant and is associated with a control barrier function $h_j$. A key observation underlying our approach is that the union of control invariant sets is itself control invariant. In particular, the aggregated safe set:
\begin{align}
\Cc_{\max} \triangleq \bigcup_{j\in \until p} \Cc_j =\setdefb{\bx\in\real^n}{h_{\max}(\bx) \geq 0},
\end{align}
where $h_{\max}(\bx)\defeq \max_{j\in\until p} h_j(\bx)$,
is also control invariant and therefore safe. Moreover, as reviewed in Sec.~\ref{sec:filter}, this aggregated safe set can be addressed using the combinatorial CBF framework. In particular, we can find a controller $\map{\bk}{\real^n}{\Uc}$
rendering the safe set $\Cc_{\max}$ forward invariant, provided that the given CBFs $\{h_j\}_{j=1}^p$ satisfy the combinatorial CBF condition, which requires, for each $\bx\in\Cc_{\max}$, the existence of $\bu\in\Uc$  satisfying:
\begin{equation}\label{eq:combo_CBF_condition}
\dot h_j(\bx,\bu) > -\alpha(h_{\max}(\bx)+|h_j(\bx)-h_{\max}(\bx)|),~\forall j\in\until{p}.
\end{equation}

The combinatorial CBF condition~\eqref{eq:combo_CBF_condition} must hold over the entire aggregated set $\Cc_{\max}$. On the other hand, each CBF only guarantees $\bu$ satisfying its respective constraint on $\Cc_j$. Consequently, its combination in~\eqref{eq:combo_CBF_condition} is not guaranteed. To address this issue, we introduce a relaxation mechanism.

\begin{lemma}\label{lem:individual-feasibility}
    Suppose that a function $\map{h_j}{\real^n}{\real}$ is a CBF for~\eqref{sys:ctrl_affine} under a given input constraint $\Uc\subseteq\real^m$. Then, for any function $\map{h}{\real^n}{\real}$ and any positive definite function $\map{\rho}{\real}{\realnonneg}$, there exists a pair of a control input $\bu\in\Uc$ and an auxiliary variable $\omega \geq 0$ such that the condition:
    \begin{equation}\label{eq:gen_combo_CBF_cond}
    \dot h_j(\bx,\bu) > -\alpha(h_j(\bx)) -\omega\rho(h_j(\bx)-h(\bx)),
    \end{equation}
    holds for any $\bx\in\Cc$ in the zero-superlevel set of $h$ in~\eqref{eq:safe_set}. 
\end{lemma}
\begin{proof}
    If ${h_j(\bx)\!\neq\! h(\bx)}$, then ${\rho(h_j(\bx)\!-\!h(\bx))}$ is strictly positive. In this case, $\omega$ can be chosen sufficiently large so that the inequality~\eqref{eq:gen_combo_CBF_cond} holds, regardless of the choice of $\bu$. 
    
    On the other hand, if ${h_j(\bx)= h(\bx)}$, then the inequality~\eqref{eq:gen_combo_CBF_cond} is equivalent to:
    ${\dot h_j(\bx,\bu) > -\alpha(h_j(\bx))}$. Since ${h_j(\bx)= h(\bx)\geq 0}$ implies that $\bx$ belongs to $\Cc_j$, the existence of a control input ${\bu\in\Uc}$ satisfying this inequality is guaranteed by the fact that $h_j$ is a CBF, regardless of the choice of $\omega$, concluding the proof.
\end{proof}
Lemma~\ref{lem:individual-feasibility} is inspired by feasibility-restoring constructions in the optimal-decay CBF framework~\cite{PO-MHC-TGM-ADA:25-cdc,JZ-BZ-ZL-KS:21}, in which auxiliary variables are introduced to relax barrier conditions while preserving safety guarantees. In our paper, the auxiliary variable $\omega$ plays an important role in scaling the relaxation term $\rho(\cdot)$ in order to recover feasibility of the combinatorial CBF condition outside the set $\Cc_j$. The introduction of $\rho$ here generalizes the absolute value appearing in the combinatorial CBF construction in~\eqref{eq:combo_CBF_condition}. Notably, placing the relaxation term outside the function $\alpha$ allows the auxiliary variable $\omega$ to enter the constraint linearly, which will be crucial for the optimization-based controller synthesis developed later.

It is important to note that while Lemma~\ref{lem:individual-feasibility} guarantees feasibility of each individual combinatorial CBF constraint for all $\bx \in \Cc_{\max}$, uniting multiple safe sets requires these constraints to be satisfied simultaneously. This observation motivates the introduction of a generalized combinatorial CBF construction, which we define next.

\begin{definition}\label{def:gen_combo_cbf}
    \longthmtitle{Generalized Combinatorial CBFs}
    A function $\map{h}{\real^n}{\real}$ constructed pointwise to take the $r$-th largest value among in a collection of functions $\{h_j\}_{j=1}^p$ as:
    \begin{equation}\label{eq:gen_combo_cbf}
    h(\bx)= \max^r \{h_j(\bx)\}_{j=1}^p,
    \end{equation}
    is a \textit{generalized combinatorial CBF ($p$-choose-$r$ CBF)} for system~\eqref{sys:ctrl_affine} under a given input constraint $\Uc\subseteq \real^m$ if there exists a class-$\Ke$ function $\alpha$ such that, for each $\bx \in \Cc$ in~\eqref{eq:safe_set}, there exists a control input $\bu\in\Uc$ and an auxiliary variable $\omega\in\realnonneg$ satisfying \eqref{eq:gen_combo_CBF_cond} simultaneously for all $j\in \until{p}$.~\hfill$\diamond$
\end{definition}

Definition~\ref{def:gen_combo_cbf} implicitly requires a compatibility property among the functions $\{h_j\}_{j=1}^p$ at each $\bx\in\Cc$, ensuring that the associated constraints can be satisfied simultaneously. However, note that the original combinatorial CBF framework in \eqref{eq:comb-CBF-QP} requires a global compatibility condition among all CBFs $\{h_j\}_{j=1}^p$ being combined. In contrast, the proposed generalized construction relaxes this compatibility requirement. The introduction of the auxiliary variable $\omega$ automatically ensures feasibility of the $j$-th inequality outside its safe region $\Cc_j$. As a result, compatibility is only required among CBFs that are \textit{active} at a given state. Consequently, the compatibility requirement reduces to the conjunctive (AND-type) compatibility of the \textit{active} CBFs on their regions of intersection. We formalize this statement below.

\begin{definition}\label{def:conj_compat}
    \longthmtitle{Conjunctive Compatibility}
    Given a collection of CBFs $\{h_j\}_{j=1}^p$ for system~\eqref{sys:ctrl_affine}, denote the set of \textit{active} indices at a state $\bx$ as
       $ \Ac(\bx) = \setdefb{j\in\until p}{h_j(\bx)\geq 0}.$
    The CBFs $\{h_j\}_{j=1}^p$ are said to be \textit{conjunctively compatible} at $\bx$ if there exists a control input $\bu\in\Uc$ such that:
    \begin{equation}\label{eq:conj_compat}
    \dot h_j(\bx,\bu) > -\alpha(h_j(\bx)),~\forall j\in\Ac(\bx).
    \end{equation}
    Furthermore, given a set $\Cc\subseteq\real^n$, the CBFs $\{h_j\}_{j=1}^p$ are \textit{conjunctively compatible} on $\Cc$ if they are conjunctively compatible at every $\bx\in\Cc$.~\hfill $\diamond$
\end{definition}
\begin{proposition}\longthmtitle{Conjunctive Compatibility Implies Generalized Combinatorial CBF}
    \label{prop:conj_compat}
    Let $h$ be constructed from CBFs $\{h_j\}_{j=1}^p$ for system~\eqref{sys:ctrl_affine} as in~\eqref{eq:gen_combo_cbf}. If CBFs $\{h_j\}_{j=1}^p$ are conjunctively compatible on the zero-superlevel set $\Cc$ of $h$ as defined in~\eqref{eq:safe_set}, then the function $h$ is a generalized combinatorial CBF. 
\end{proposition}
\begin{proof}
    From conjunctive compatibility, there exists, for each $\bx \in \Cc$, a control $\bu\in\Uc$ satisfying~\eqref{eq:conj_compat}. For such a $\bu$, the inequalities~\eqref{eq:gen_combo_CBF_cond} are satisfied for all $j\in\Ac(\bx)$ regardless of the choice of $\omega\geq 0$, since the relaxation term $-\omega\rho(\cdot)$ is nonpositive. It remains to show that there exists $\omega$ sufficiently large such that, with this choice of $\bu$, the inequalities~\eqref{eq:gen_combo_CBF_cond} also hold for all $j\not\in \Ac(\bx)$.
    
    For indices $j\not\in \Ac(\bx)$, the definition of $\Ac$ implies $h_j(\bx)< 0 \leq h(\bx)$, since $\bx\in\Cc$. Therefore, the expression $\rho(h_j(\bx)-h(\bx)) >0$ is strictly positive for each $j\notin \Ac(\bx)$. Then, because there are a finite number of inequalities, a single sufficiently large $\omega$ can be chosen so that the inequalities~\eqref{eq:gen_combo_CBF_cond} hold for all $j\not\in \Ac(\bx)$. Hence, this pair $(\bu,\omega)$ satisfies inequalities~\eqref{eq:gen_combo_CBF_cond} for all $j\in\until p$, concluding the proof.
\end{proof}
\begin{remark}\label{rmk:compat}
\longthmtitle{On the Notion of Compatibility}
    Definition~\ref{def:conj_compat} can be further refined. First, the notion of conjunctive compatibility can be tightened by defining the active index set as ${\Ic(\bx) = \setdefb{j \in \until{p}}{h_j(\bx) = h(\bx)}}$, in which case Proposition~\ref{prop:conj_compat} continues to hold. Moreover, ideas from the optimal-decay CBF framework~\cite{PO-MHC-TGM-ADA:25-cdc,JZ-BZ-ZL-KS:21} can be incorporated to further generalize the compatibility notion. A thorough investigation of these extensions is beyond the scope of this paper. For clarity of exposition, we adopt the standard conjunctive notion of compatibility in this paper.~\hfill$\bullet$    
\end{remark}
We now show that the generalized combinatorial CBF condition yields a safety guarantee when enforced through an optimization-based controller.

\begin{theorem}\label{thm:gen_combo_cbf_safety}\longthmtitle{Safety from Generalized Combinatorial CBF}
    Consider the control-affine system~\eqref{sys:ctrl_affine}. Let $h$ be constructed from CBFs $\{h_j\}_{j=1}^p$ for system~\eqref{sys:ctrl_affine} as in~\eqref{eq:gen_combo_cbf}. If $h$ is a generalized combinatorial CBF for~\eqref{sys:ctrl_affine}, then its associated zero-superlevel set~$\Cc$, as defined in~\eqref{eq:safe_set}, is control invariant.

    In particular, consider the CBF-QP given by:
    \begin{align}\label{eq:gen-comb-CBF-QP}
    &\bk(\bx)  =\argmin_{\bu\in\Uc,~\omega\geq 0} \quad  \|\bu-\bk_\des(\bx)\|^2+c_\omega\omega^2\\
     &\textup{s.t.} \quad  \dot h_j(\bx,\bu) \geq -\alpha(h_j(\bx)) -\omega\rho(h_j(\bx)-h(\bx)),~\forall j\in\until p. \nonumber
    \end{align}
    with a weight ${c_\omega>0}$. Under the above assumptions, the CBF-QP is continuous at each ${\bx\in\Cc}$, and the set $\Cc$ is forward invariant for the closed-loop system under ${\bu=\bk(\bx)}$.
\end{theorem}
\begin{proof}
    Our forward invariance proof relies on showing that $h$ satisfies the nonsmooth barrier function condition for the closed-loop system~\cite[Prop. 2]{PG-JC-ME:17}. To this end, we first note that, by Definition~\ref{def:gen_combo_cbf}, the generalized combinatorial CBF condition is imposed with strict inequalities.
    As such, continuity of the functions $h_j$ and their Lie derivatives ensures strict feasibility persists on an open neighborhood $\Ec\supset\Cc$. Hence, the CBF-QP remains feasible and well-defined on $\Ec$.
    
    Since $h$ is constructed via sorting continuously differentiable functions $h_j$, it is nonsmooth, yet it is locally Lipschitz. Therefore, $h$ admits a Clarke generalized gradient~\cite{FHC:83} at nonsmooth points. The generalized gradient of $h$ at $\bx$ is the set
    $
    \partial h(\bx) = \textrm{conv}  \cup_{j\in\Ic(\bx)} \left\{\nabla h_j(\bx)\right\}
    $
    where: 
    $$
    \Ic(\bx) = \setdefb{j\in\until p}{h_j(\bx)=h(\bx)}
    $$
    is the set of indices of $h_j$ with the same value as~$h$.

    Under ${\bu=\bk(\bx)}$, the QP constraints are satisfied for all ${j\in \until p}$. For all ${j\!\in\! \Ic(\bx)}$, we have ${h_j(\bx)\!-\!h(\bx)\!=\!0}$, and thus:
    $$
    \langle \nabla h_j(\bx), \bf(\bx)+\bg(\bx)\bk(\bx) \rangle \geq - \alpha(h(\bx)).
    $$
    Then since $\partial h(\bx)$ is the convex hull and the inner product is linear, the same lower bound as the above holds:
    $$
    \langle \xi, \bf(\bx)+\bg(\bx)\bk(\bx) \rangle \geq - \alpha(h(\bx)).
    $$
    for all ${\xi \in \partial h(\bx)}$ and all ${\bx\in \Ec}$. Thus, the function $h$ satisfies the nonsmooth barrier function condition on a neighborhood ${\Ec\supset\Cc}$, guaranteeing that $\Cc$ is forward invariant for the closed-loop system~\cite[Prop. 2]{PG-JC-ME:17}. Control invariance of  system~\eqref{sys:ctrl_affine} directly follows from the existence of control signal $\bu(t)=\bk(\bphi(t,\bx_0))$ for any $\bx_0\in\Cc$ satisfying Definition \ref{def:contr-invariance}.

    In addition, Definition~\ref{def:gen_combo_cbf} ensures the existence of $\bu$ strictly satisfying the constraints in the CBF-QP for each $\bx\in\Cc$. In other words, at each $\bx\in\Cc$, the CBF-QP satisfies Slater's condition, so it is continuous~\cite{PM-AA-JC:25} on $\Cc$. 
\end{proof}

The proposed generalized combinatorial CBF framework provides a less restrictive mechanism for combining multiple CBFs. In particular, within the CBF-QP framework, the auxiliary variable $\omega$ is computed online, as an additional decision variable in~\eqref{eq:gen-comb-CBF-QP}, automatically restoring feasibility whenever necessary. Theorem~\ref{thm:gen_combo_cbf_safety} establishes safety guarantees  for the general $p$-choose-$r$ constructions considered in~\cite{PO-HL-TGM-DP-ADA:25-csl}, enabling flexible logical combination of safe sets.

For the purpose of this paper, however, we focus primarily on the disjunctive (OR) case, corresponding to~${r=1}$, which represents the union of safe sets. In this setting, the combinatorial CBF construction offers a natural way to unite multiple certified safe regions under a single continuous controller, without requiring explicit blending or switching among individual feedback laws.
\begin{corollary}\label{cor:unite_safe_sets}\longthmtitle{Uniting Multiple Safe Sets}
    Consider the control-affine system~\eqref{sys:ctrl_affine}. Given multiple sets $\{\Cc_j\}_{j=1}^p$ with safety verified by associated CBFs $\{h_j\}_{j=1}^p$, the union of the safe sets $\Cc_{\max} = \cup_{j=1}^p\Cc_j$ is safe.

    In particular, under the assumption that the CBFs $\{h_j\}_{j=1}^p$ are conjunctively compatible, the CBF-QP~\eqref{eq:gen-comb-CBF-QP} is continuous at each $\bx\in\Cc$, and the set $\Cc$ is forward invariant for the closed-loop system under $\bu=\bk(\bx)$.~\hfill$\blacksquare$
\end{corollary}
The first statement of Corollary~\ref{cor:unite_safe_sets} follows directly from the fact that the union of control invariant sets is itself control invariant. The CBF-QP~\eqref{eq:gen-comb-CBF-QP}, however, provides a constructive safety filter that remains close to the nominal controller. An advantage of this approach is that it preserves continuity of the resulting feedback law, thereby avoiding the need for explicit switching, blending, or hybrid control analysis when transitioning between safe regions. The conjunctive compatibility assumption serves as a sufficient condition under which such a continuous controller can be constructed over the aggregated safe set (see also Remark~\ref{rmk:compat}).

\section{Uniting Multiple Implicit Safe Sets}\label{sec:multi-implicit-set}
In this section, we extend the generalized combinatorial CBF framework developed in Sec.~\ref{sec:gen-combo-cbf} to implicit safe sets arising from multiple backup controllers. While a single backup construction expands a certified region, it may still yield a conservative operating set. Our objective is to aggregate multiple implicit safe sets into a larger certified region enforceable through a single optimization-based safety filter.

\begin{remark}\longthmtitle{On Multiple Safe Sets}
In many applications, it is easier to certify several \emph{small} safe sets than to directly construct a single large one. For example, stabilizing multiple equilibria, each with a distinct controller, is often simpler than identifying a single equilibrium that yields a large safe Lyapunov sublevel set. The results of this section provide a principled mechanism to aggregate them.
\end{remark}

Suppose multiple discretized implicit safe sets~$\{\Cc_{{\rm I},j}^{\rm fin}\}_{j=1}^p$ are constructed from CBFs $\{h_j\}_{j=1}^p$ with associated backup controllers $\{\bk_{{\rm b},j}\}_{j=1}^p$ as in~\eqref{eq:implicit_safe_set_finite}. While the results of Sec.~\ref{sec:gen-combo-cbf} establish a mechanism for uniting multiple explicitly defined safe sets, they are not directly applicable in the implicit case. In particular, a naive approach would be to first unite the explicit safe sets certified by CBFs $\{h_j\}_{j=1}^p$ using our approach introduced in Sec.~\ref{sec:gen-combo-cbf}, and then apply the backup CBF framework of Sec.~\ref{sec:vanillaBackup} to expand the resulting aggregated safe set. However,  the backup CBF construction requires a continuously differentiable safeguarding (backup) controller in order to define the associated flow and evaluate the trajectory-level Lie derivatives. In contrast, the controller produced by the CBF-QP~\eqref{eq:gen-comb-CBF-QP} is, in general, only continuous and may be nonsmooth.

Our approach directly combines the implicit safe sets $\{\Cc_{{\rm I},j}^{\rm fin}\}_{j=1}^p$ to obtain an aggregated implicit safe set:
\begin{align}\label{eq:aggregated_implicit_safe_set-finite}
\Cc_{\rm agg}^{\rm fin} &= \bigcup_{j=1}^p \Cc_{{\rm I},j}^{\rm fin} = \bigcup_{j=1}^p \setdefb{\bx\in\real^n}{h_{{\rm I},j}^{\rm fin}(\bx)\geq 0} \nonumber,\\
&= \setdefb{\bx \in \real^n}{h_{\rm agg}^{\rm fin}(\bx)\defeq \max_{j\in\until p} h_{{\rm I},j}^{\rm fin}(\bx) \geq 0}.
\end{align}
Recall here that each $h_{{\rm I},j}^{\rm fin}$ is itself defined as a minimum over trajectory constraints, cf.~\eqref{eq:implicit_CBF_finite}. Consequently, the aggregated function~$h_{\rm agg}^{\rm fin}$ exhibits a nested max-min structure, corresponding to a disjunctive logical combination (OR) applied to conjunctive logical combinations (AND). Nested logical compositions of CBFs were studied in~\cite{PO-HL-TGM-DP-ADA:25-csl}. In general, a two-level nested composition can be expressed through sorting operations of the form:
\begin{equation}
    h(\bx)=\max^{r_2} \left\{\max^{r_{1,j}} \{h_{j,k}(\bx)\}_{k=1}^{p_{1,j}}\right\}_{j=1}^{p_2},
\end{equation}
where for each outer index $j\in\until{p_2}$, the functions $\{h_{j,k}\}_{k=1}^{p_{1,j}}$ define a collection of $p_{1,j}$ barrier functions whose $r_{1,j}$-th order statistic encodes an inner logical combination. The outer sorting over $j$ then selects the $r_2$-th order statistic across the resulting groups.

Although this nested structure may appear intricate, it can be handled within the generalized combinatorial CBF framework developed in Sec.~\ref{sec:gen-combo-cbf}. In particular, the CBF-QP~\eqref{eq:gen-comb-CBF-QP} with constraints on functions $h_{j,k}$ can render the zero-superlevel set $\Cc$ associated with the function $h$ safe, by establishing that $h$ is a nonsmooth barrier function. Further, the introduction of auxiliary variables similarly restores feasibility under a conjunctive compatibility assumption among the primitive CBFs $h_{j,k}$. We omit the formalization of these results due to space limitations and instead focus on their specialization to the problem of combining backup CBFs. 

\begin{definition}\label{def:agg_implicit_cbf}
    \longthmtitle{Aggregated Implicit CBF - discretized version}
    A function ${\map{h_{\rm agg}^{\rm fin}}{\real^n}{\real}}$ defined as in \eqref{eq:aggregated_implicit_safe_set-finite}, i.e., as the pointwise maximum of the discretized implicit CBFs $\{h_{{\rm I},j}^{\rm fin}\}_{j=1}^p$, is called an \textit{aggregated implicit CBF} for system~\eqref{sys:ctrl_affine} under a given input constraint $\Uc\subseteq \real^m$ if there exists a class-$\Ke$ function $\alpha$ and a positive definite function~$\rho$ such that, for each $\bx \in \Cc_{\rm agg}^{\rm fin}$ in~\eqref{eq:aggregated_implicit_safe_set-finite}, there exists a control input $\bu\in\Uc$ and an auxiliary variable $\omega\in\realnonneg$ satisfying:
    \begin{subequations}\label{eq:gen-combo-implicit-CBF-finite-condition}
        \begin{align}
            &\dot \psi_{\bk_{\rm b},j}(\tau_k,\bx,\bu) > -\alpha(\psi(\bphi_{\bk_{\rm b},j}(\tau_k, \bx)))\\
            &\qquad\qquad\qquad-\omega\rho\left(\psi(\bphi_{\bk_{\rm b},j}(\tau_k, \bx))-h_{\rm agg}^{\rm fin}(\bx)\right),\nonumber\\
            &\dot h_{\bk_{\rm b},j}(T,\bx,\bu) > -\alpha(h_j(\bphi_{\bk_{\rm b},j}(T, \bx)))\\
            &\qquad\qquad\qquad-\omega\rho\left(h_j(\bphi_{\bk_{\rm b},j}(T, \bx))-h_{\rm agg}^{\rm fin}(\bx)\right). \nonumber
        \end{align}
    \end{subequations}
    simultaneously for all timestep $k\in \until{N}$ and all backup controllers indexed by $j\in \until{p}$.~\hfill$\diamond$
\end{definition}

The function $\rho$ in Definition~\ref{def:agg_implicit_cbf} may be chosen as any positive definite function, and serves to scale the relaxation term that restores feasibility of the nested barrier constraints. Similar to the generalized combinatorial CBF framework developed in Sec.~\ref{sec:gen-combo-cbf}, the above definition implicitly requires a compatibility condition among the implicit CBFs. This implies that there exists a common backup controller that renders the sets $\{\mathcal{C}_j\}_{j=1}^p$ forward invariant. For brevity, we do not restate the notion of conjunctive compatibility in this setting and instead proceed directly to the main safety result.

\begin{theorem}\label{thm:unite_implicit}
\longthmtitle{Uniting Multiple Implicit Safe Sets}
Consider the control-affine system~\eqref{sys:ctrl_affine}. Let $h_{\rm agg}^{\rm fin}$ defined as in \eqref{eq:aggregated_implicit_safe_set-finite} be an aggregated implicit CBF for system~\eqref{sys:ctrl_affine} under a given input constraint $\Uc\subseteq \real^m$. Then, consider the CBF-QP:
    \begin{align}\label{eq:gen-combo-implicit-CBF-QP-finite}
    \bk(\bx)  &=\argmin_{\bu\in\Uc,~\omega\geq 0} \quad  \|\bu-\bk_\des(\bx)\|^2+c_\omega\omega^2\\
        \textup{s.t.}~&\quad \textup{Constraints~\eqref{eq:gen-combo-implicit-CBF-finite-condition} with non-strict inequalities} \nonumber
    \end{align}
The constraints are enforced for all timestep $k\in \until{N}$ and for all backup controller indexed by $j\in \until{p}$. The CBF-QP~\eqref{eq:gen-combo-implicit-CBF-QP-finite} is continuous at each $\bx\in\Cc_{\rm agg}^{\rm fin}$, and the set $\Cc_{\rm agg}^{\rm fin}$ is forward invariant for the closed-loop system under $\bu=\bk(\bx)$.
\end{theorem}
\begin{proof}
    The proof follows the same structure as that of Theorem~\ref{thm:gen_combo_cbf_safety}. We consider the index sets:
    \begin{align*}
    \Ic_\psi(\bx) &= \setdefb{(j,k)\in\until p \times \until N}{\psi(\bphi_{\bk_{\rm b},j}(\tau_k, \bx))=h_{\rm agg}^{\rm fin}(\bx)},\\
    \Ic_h(\bx) &= \setdefb{j\in\until p }{h_j(\bphi_{\bk_{\rm b},j}(T, \bx))=h_{\rm agg}^{\rm fin}(\bx)}.
    \end{align*}
    At each $\bx\in\Ec$, where $\Ec$ is a neighborhood of $\Cc_{\rm agg}^{\rm fin}$ where the CBF-QP remains strictly feasible, the constraints enforced by the CBF-QP reduce to:
    \begin{align*}
    &\dot \psi_{\bk_{\rm b},j}(\tau_k,\bx,\bk(\bx)) \geq -\alpha(h_{\rm agg}^{\rm fin}(\bx)),\\
    &\dot h_{\bk_{\rm b},i}(T,\bx,\bk(\bx)) \geq -\alpha(h_{\rm agg}^{\rm fin}(\bx)),
    \end{align*}
    for each $(j,k)\in\Ic_\psi(\bx)$ and $i\in\Ic_h(\bx)$. The common lower bound ensures for all $\xi \in \partial h_{\rm agg}^{\rm fin}(\bx)$:
    $$
    \langle \xi, \bf(\bx)+\bg(\bx)\bk(\bx) \rangle \geq - \alpha(h_{\rm agg}^{\rm fin}(\bx)).
    $$
    This establishes $h_{\rm agg}^{\rm fin}$ as a nonsmooth barrier function for the closed-loop system, concluding the proof.
\end{proof}
Theorem~\ref{thm:unite_implicit} establishes that multiple discretized implicit safe sets, each generated by a distinct backup controller, can be aggregated through a single optimization-based safety filter while preserving continuity and forward invariance. This accomplishes our objective of uniting certified safe operating regions into a larger aggregated safe set.

\section{Applications to Space Systems}
\subsection{Spacecraft attitude}
Consider an underactuated, rotating satellite, modeled by
\begin{align}
    \dot \bR &= \bR \hat \bOmega, \; \bJ\dot \bOmega + \bOmega \times \bJ \bOmega = \bB \bu. \label{eq:satellite-orient}
\end{align}
Here ${\bR \in SO(3)}$ is the orientation of the satellite with respect to a fixed frame, $\bOmega \in \real^3$ its body angular velocity, and $\hat \bOmega \in \real^{3\times 3}$ is the unique skew matrix for which $\hat \bOmega \bv = \bOmega \times \bv$ for all $v \in \R^3$. The satellite is further rotationally symmetric about the $\be_3$-axis, having $\bJ = \diag(\lambda, \lambda, \hat \lambda)$. The actuation matrix for the system is $\bB = [\be_1, \be_2]$, where $\be_i$ is the $i$'th standard basis vector in $\real^3$. Here, $\bu \in \real^2$ must be chosen such that $\lVert \bu \rVert \leq u_{\max}$, for small $u_{\max} > 0$. 

We consider a safety problem analogous to that of \cite{MD-PO-ADA:26-acc}. In order to stay protected from the Sun, the satellite must orient its heat shield, which is normal to the body-fixed $\be_3$-axis, within a safe angle $\theta_{\text{safe}}$ of the spatial $\be_3$-axis. This imposes a state constraint of the form:
\begin{align*}
    \Sc = \setdefb{(\bR, \bOmega)}{\psi(\bR, \bOmega) = \be_3^\top \bR \be_3 - \cos(\theta_{\text{safe}}) \geq 0}.
\end{align*}
Since $\Sc$ depends only on $\bR\be_3$, the problem is simplified by first reducing the satellite to a system on the sphere $\mathbb S^2$, via the projection $\bR \mapsto \bR\be_3$. By standard symmetry reduction techniques, \eqref{eq:satellite-orient} reduces to a fully-actuated system on $\mathbb S^2$. On $\mathbb S^2$, five backup sets $\{\Cc_j\}_{j = 1}^5$ are constructed. Each set is a sublevel set of a Lyapunov function for the reduced system on $\mathbb S^2$, derived from a geometric PD backup controller $\bk_{{\rm b},j}$ stabilizing to one of $e_3$, $R_x(\pm\theta_{\text{safe}}/2)e_3$, or $R_y(\pm \theta_{\text{safe}}/2)e_3$ in $\mathbb S^2$ \cite[p. 533]{bullo2005geometric}. 

\begin{figure}
    \centering
    \includegraphics[width=0.95\linewidth]{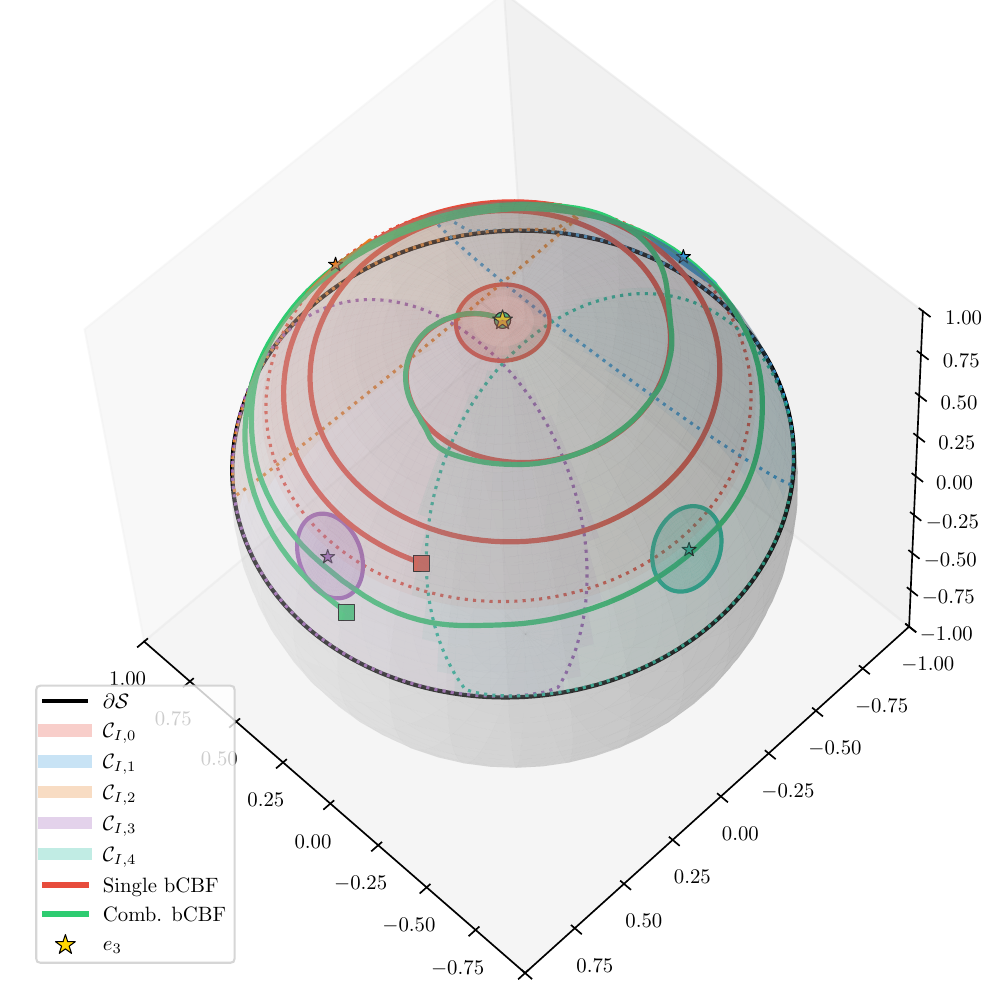}
    \vspace{-.2cm}
    \caption{Simulation results for the satellite are depicted on the sphere under the projection $\bR \mapsto \bR\be_3$. The satellite is commanded to track a trajectory circling around the border of the safe set. Five backup sets are used, with the level set of each $h_j$ plotted as a solid circle and its implicit safe set as a dotted line. Using the proposed combinatorial bCBF, plotted in green, both safety and convergence to the boundary are achieved. As evidenced by the red trajectory, convergence is not possible with a single backup set.}
    \label{fig:backup-satellite}
    \vskip -4mm
\end{figure}

In Figure \ref{fig:backup-satellite}, the CBF-QP of Theorem \ref{thm:unite_implicit} is implemented\footnote{Parameters are $u_{\rm max}\!=\!0.5$, $\lambda\!=\!0.5$, $\hat \lambda\!=\!1$, $\theta_{\rm safe}\!=\!80^\circ$, and $T\! =\! 4$. } with a nominal PD controller designed to track a trajectory oscillating in and out of the safe set. The use of multiple backup sets enables simultaneous tracking and safety, which is not achievable with a single backup set.

\subsection{Safe Spacecraft Station Keeping}

\begin{figure}[t]
    \centering
    \includegraphics[width=.49\linewidth]{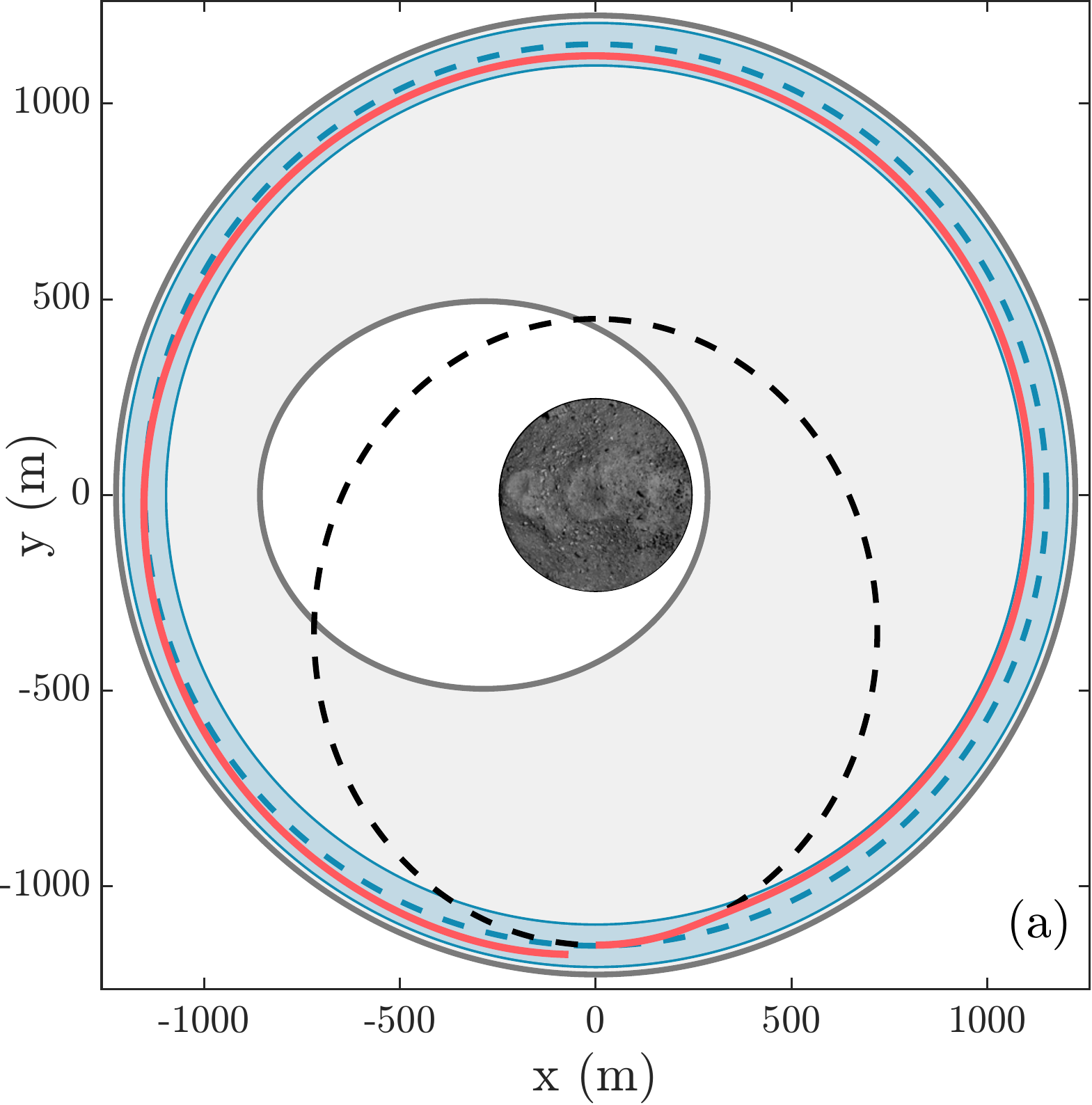}
    \includegraphics[width=.49\linewidth]{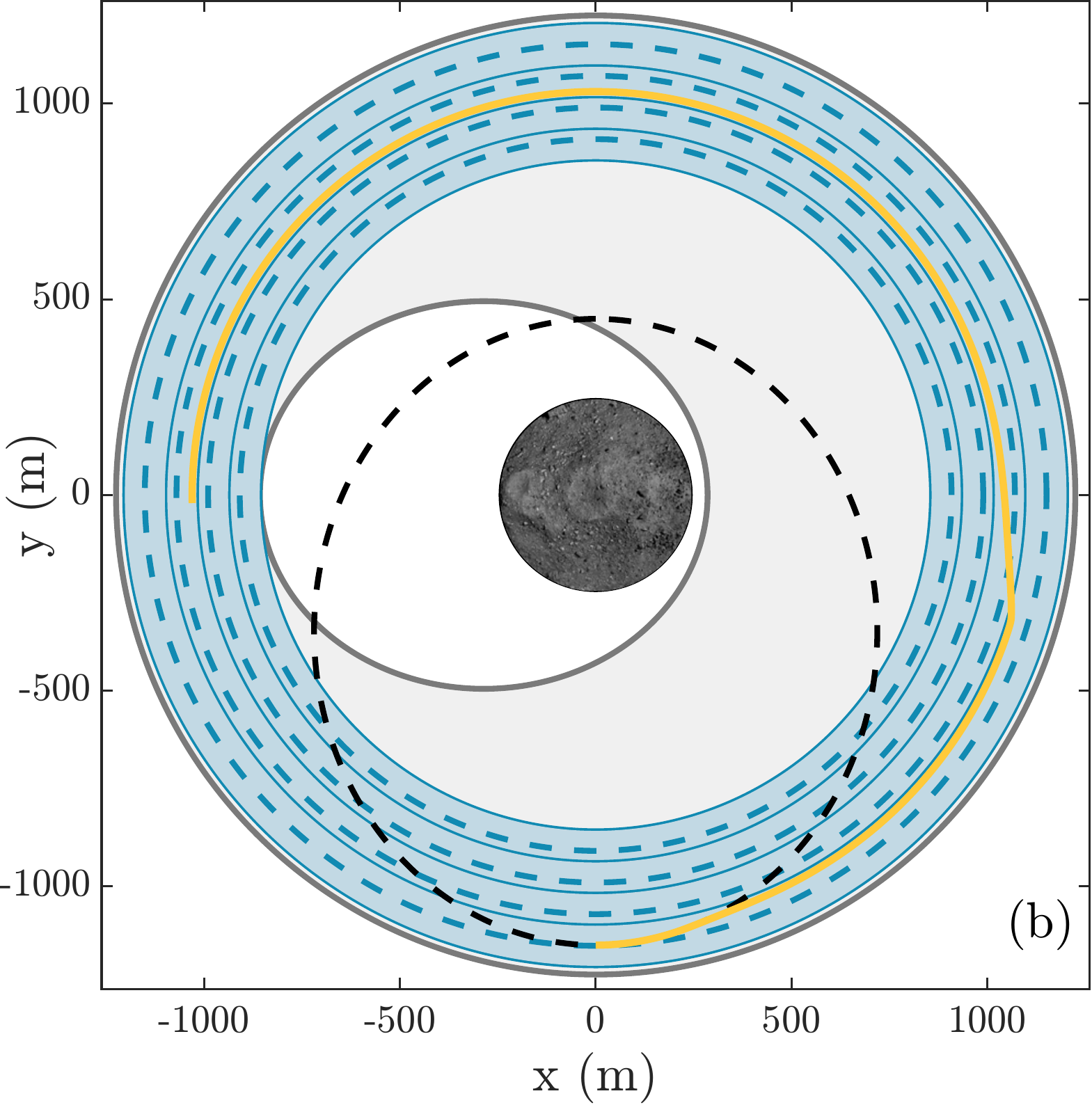}
    \includegraphics[width=.49\linewidth]{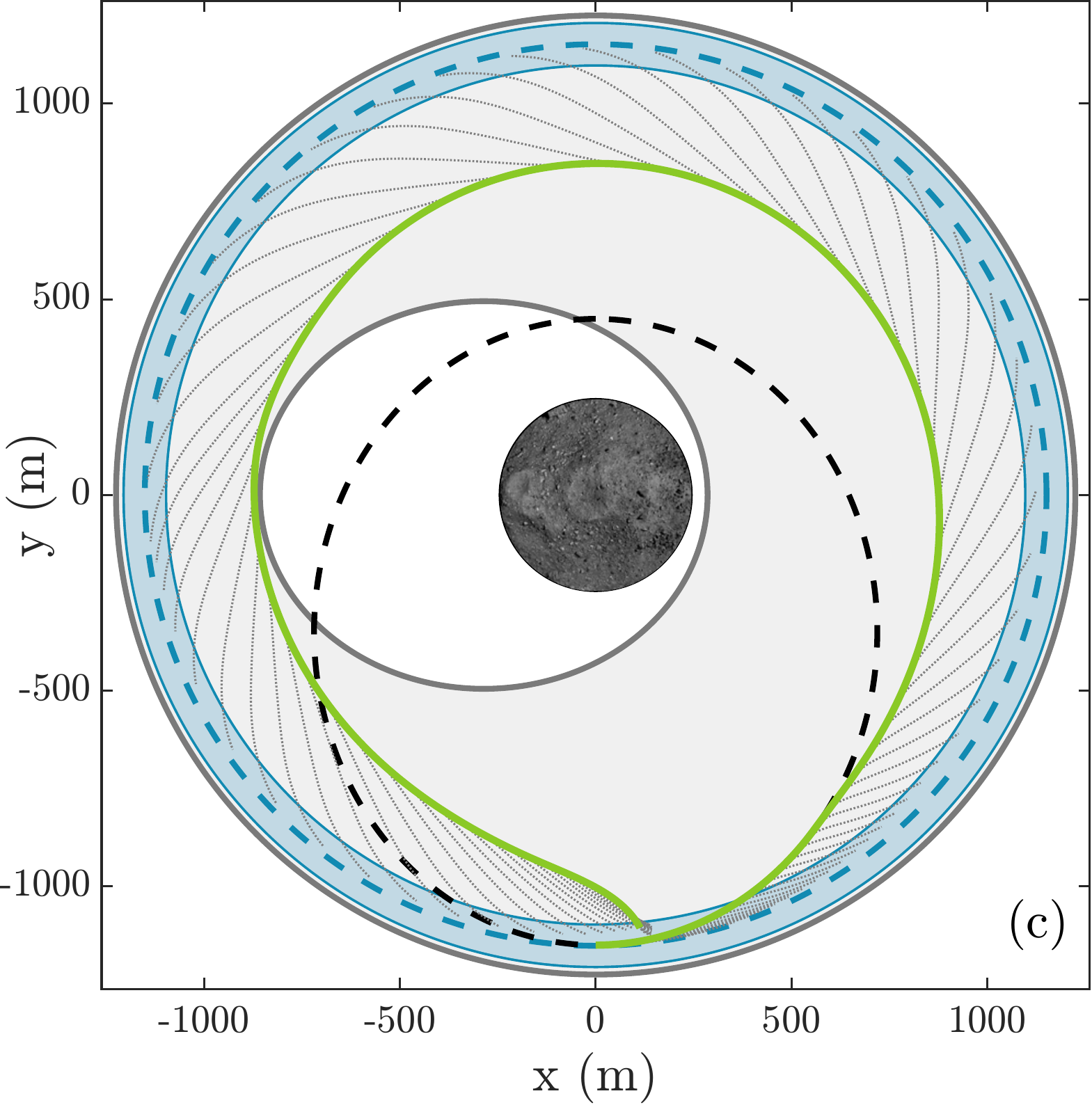}
    \includegraphics[width=.49\linewidth]{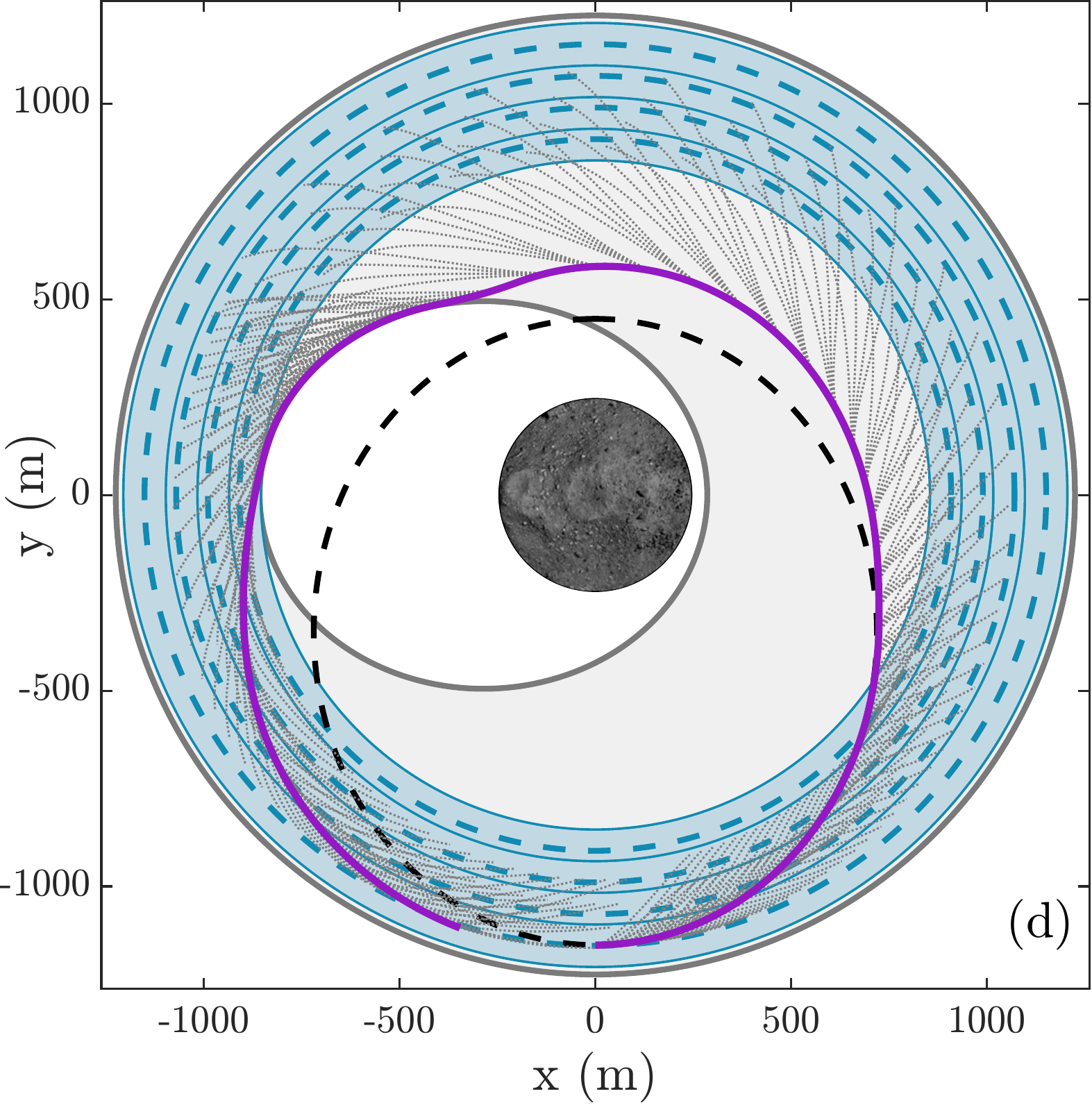}
    \includegraphics[width=.49\linewidth]{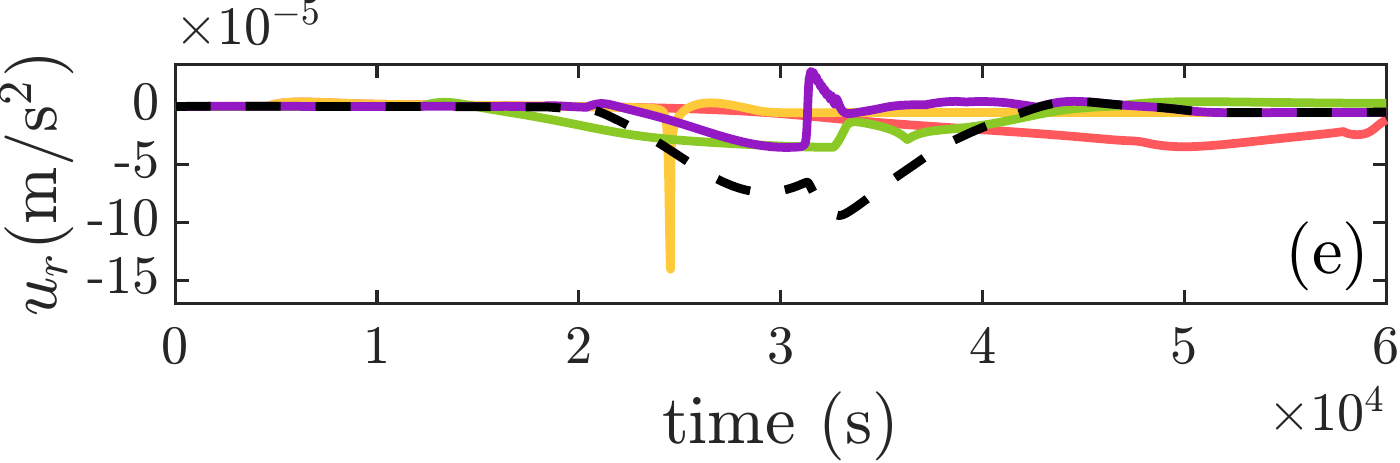}
    \includegraphics[width=.49\linewidth]{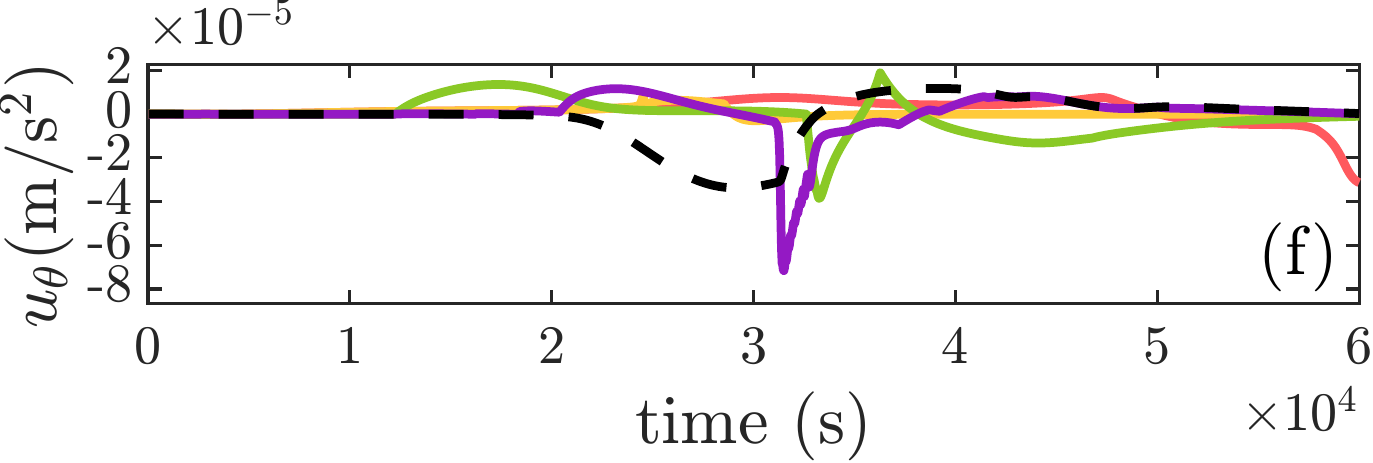}
    \vspace{.4cm}
    \hspace{.3cm}
    \includegraphics[width=.95\linewidth]{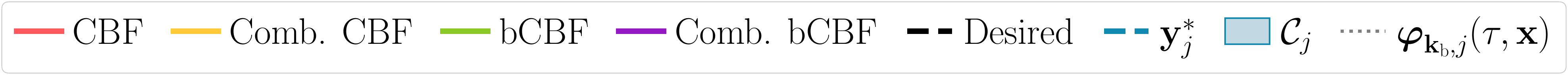}
    \vspace{-.8cm}
    \caption{Simulation results for safe station keeping comparing the standard CBF (\textbf{a}), the standard bCBF (\textbf{b}), the proposed combinatorial CBF (\textbf{c}) and the proposed combinatorial bCBF (\textbf{d}). While all four approaches guarantee the safety of the spacecraft by adhering to the keep-in and keep-out constraints (\textbf{a}-\textbf{d}), and obeying the input bounds (\textbf{e},$\nspace{2}$\textbf{f}) they vary in mission performance. 
    By using a CBF without expansion (\textbf{a},$\nspace{2}$\textbf{b}), the spacecraft cannot approach the desired orbit ({dashed black}), though using multiple CBFs without expansion allows a closer approach (\textbf{b}). Expanding a single CBF with the backup method allows for improved orbit tracking (\textbf{c}), but expanding with multiple backup controllers and CBFs using the proposed approach yields the largest control invariant safe set, and thus achieves superior orbit tracking (\textbf{d}).
    }
    \label{fig:stationkeeping}
    \vskip - 4mm
\end{figure}

Consider the scenario of a satellite orbiting an asteroid in a fixed plane, in order to acquire surface and feature information for a future probe mission. The planar satellite dynamics in polar coordinates can be described by:
\begin{equation} \label{eq:satellite_dyn}
    {\small\begin{bmatrix}
         \dot r \\ \dot \theta \\ \ddot r \\ \ddot \theta
    \end{bmatrix}} = {\small\begin{bmatrix}
        \dot r \\
        \dot \theta \\
        r\dot \theta^2 - \frac{\mu}{r^2}\\
        -\frac{2}{r}\dot r \dot \theta
    \end{bmatrix}}
    +{\small\begin{bmatrix}
        0 & 0 \\
        0 & 0 \\
        1 & 0 \\
        0 & \frac{1}{r}
    \end{bmatrix}}
    \bu,
\end{equation}
with states ${\bx =[
    r \nspace{10} \theta \nspace{10} \dot r \nspace{10} \dot \theta
]^\top}$ denoting the radial position and angle with respect to the asteroid, and their time rates of change, respectively. The satellite is assumed to have a continuous low-thrust electric propulsion system 
where $\boldsymbol{u} \in \mathcal{U} = [-u_{\rm max}, u_{\rm max}]^2$. In \eqref{eq:satellite_dyn}, $\mu = G M$ is the standard gravitational parameter for the asteroid, where $G$ is the gravitational constant and $M$ is the mass of the asteroid\footnote{The simulation uses the parameters for the asteroid 101955 Bennu.}.

In this scenario, upon entering an orbit near the asteroid, the satellite detects a debris field to be avoided that is characterized by an elliptical region encircling the asteroid. To avoid the region, one of the state constraints is: 
\begin{align*}
    \Sc_1 = \setdefB{\bx\in\real^n}{\psi_1(\bx) = r - \frac{p_\obs}{1+e_\obs\cos \theta} \geq 0},
\end{align*}
where ${p_\obs > 0}$ and ${e_\obs \in [0,1)}$ are the semi-latus rectum and the eccentricity of the outer edge of the debris field, respectively. Further, the satellite must remain within a region where high-quality scientific data can be obtained: 
\begin{align*}
        \Sc_2 = \setdefb{\bx\in\real^n}{\psi_2(\bx) = R - r \geq 0},
\end{align*}
where ${R > 0}$ must be chosen such that ${p_\obs \leq R (1-e_{\rm o})}$. The desired controller ${\bk_\des : \R^n \rightarrow \mathcal{U}}$ tracks an elliptical orbit described by ${r = p_{\rm d}/(1 + e_{\rm d}\cos(\theta))}$ for desired semi-latus rectum and eccentricity ${p_{\rm d}>0}$ and ${e_{\rm d}\in[0,1)}$, respectively. 

The backup sets $\{\mathcal{C}_j\}_{j=1}^p$  are defined as sublevel sets of Lyapunov functions centered around circular orbits which are described by ${\by_j^* \triangleq (r^*, \dot r^*, \dot \theta^*) = \left(r^*, 0, \sqrt{\frac{\mu}{{r^*}^3}}\right)}$, such that: 
\begin{align*}
    \mathcal{C}_j \!=\! \setdefb{\bx\in\real^n}{h_j(\bx) = \gamma_j - (\by-\by_j^*)^\top\bP_j(\by-\by_j^*) \geq 0},
\end{align*}
for ${\by \!=\! [r \nspace{10} \dot{r} \nspace{10} \dot{\theta}]^\top}$ and ${\gamma_j \!> \!0}$. Here, ${\bP_j \!=\!\bP^\top_j \!\succ\! 0}$ is obtained by solving the continuous algebraic Ricatti equation\footnote{This is given by ${\bA^\top\bP +  \bP\bA - \bP\bB\bR^{-1}\bB^\top\bP + \bQ = \mathbf{0}}$.} with: 
\begin{align*}
    \bA \!=\! {\footnotesize \begin{bmatrix}
        0 & 1 & 0\\
        0 & 0 & 0\\
        0 & 0 & 0\\
    \end{bmatrix}}\!, \nspace{5}
    \bB \!=\! {\footnotesize \begin{bmatrix}
        0 & 0 \\
        1 & 0 \\
        0 & 1
    \end{bmatrix}}\!, \nspace{5}
    \bQ\!=\! \bQ^\top \!\succ\! 0,
    \nspace{5}
    \bR\! =\! \bR^\top \!\succ\! 0,
\end{align*}
after \eqref{eq:satellite_dyn} is transformed via feedback linearization \cite{LG-AKK-TGM:25}.

The backup controllers $\{\bk_{{\rm b},j}\}_{j=1}^p$ stabilize \eqref{eq:satellite_dyn} to their respective orbits $\by^*_j$ via Sontag's Universal Formula for stabilization \cite{EDS:89a} which is saturated to obey the input bounds. 
Thus for all ${j \!\in  \!\until p}$, $\gamma_j$ must be selected such that the backup controller $\bk_{{\rm b},j}$ does not saturate within $\mathcal{C}_j$ and ${\mathcal{C}_j \!\subset \!\Sc_1 \!\cap\! \Sc_2}$.

\Cref{fig:stationkeeping} plots the simulation results\footnote{The simulation uses the constants ${p = 4}$, ${u_{\rm max} = 2.5\!\times\!10^{-4}\nspace{3}{\rm m/s^2}}$, 
${M = 7.329\!\times\!10^{10} \nspace{3}{\rm kg}}$, ${G \!=\! 6.674\!\times\!10^{-11} \nspace{3}{\rm m^3 kg^{-1} s^{-2} }}$, ${T \!=\! 5\!\times\!10^3 \nspace{3}{\rm s}}$,
${R \!= \!1.225\!\times\!10^3\nspace{3}{\rm m}}$,
${p_{\rm o} \!=\! 428.8 \nspace{3}{\rm m}}$, ${e_{\rm o} \!=\! 0.5}$, ${p_{\rm d}\! =\! 646.9 \nspace{3}{\rm m}}$, ${e_{\rm d} \!= \!0.4375}$. In nondimensional units: ${\gamma_j = 0.05 \nspace{5} \forall j \in \until{p}}$, ${\bQ = \bI_3}$, ${\bR = \bI_2}$.}$^{\nspace{-2},}$\footnote{For numerical stability, the states and control signals are nondimensionalized during computation with characteristic length given by the mean radius of the asteroid ($245.03 \nspace{3}{\rm m}$) and characteristic time selected such that ${\mu = 1}$ in dimensionless units.} for \eqref{eq:satellite_dyn} comparing the standard CBF approach, the standard backup CBF approach, the combinatorial CBF approach, and the combinatorial backup CBF approach. Due to the tight input constraints, the safe sets $\mathcal{C}_j$ are small compared to the constraint set defined by ${\Sc_1 \cap \Sc_2}$. Therefore, when using a single CBF without expansion in \eqref{eq:CBF-QP}, the motion of the spacecraft is restricted by the safety condition. By combining multiple CBFs using \eqref{eq:gen-comb-CBF-QP}, the conservatism is reduced by traveling from the outer-most safe set to inner safe sets that are closer to the desired orbit. The conservatism is further reduced via expansion of a single safe set, by enforcing the reachability of $\mathcal{C}$ under the controller $\bk_{\rm b}$ (e.g., the trajectories in dashed gray must reach $\mathcal{C}$ in a finite time). Finally, by expanding \textit{multiple} safe sets using \eqref{eq:gen-combo-implicit-CBF-QP-finite} the spacecraft can track the desired elliptical orbit for a significant portion of the mission, adheres to the constraints defined by $\psi_1$ and $\psi_2$, and obeys the input constraints imposed by $\mathcal{U}$.

\section{Conclusion}

This paper developed a generalized combinatorial control barrier function framework for uniting multiple certified safe sets under input constraints. By leveraging the combinatorial CBF framework to combine backup CBF constructions, we showed that implicit safe sets generated by distinct backup controllers can be aggregated through a single optimization-based safety filter while preserving continuity and forward invariance. This is enabled by an auxiliary-variable-based relaxation that restores feasibility of the resulting constraints, providing a principled mechanism for expanding certified safe operating regions from locally constructed safe sets. Future work will investigate trajectory generation methods from orbital mechanics, such as Lambert's problem, to efficiently construct backup trajectories and further enhance the applicability of the framework in space systems.

\end{spacing}

\vskip - 4mm
\begin{spacing}{0.94}
\bibliographystyle{ieeetr}
\bibliography{
    bib/alias,
    bib/PO,
    bib/main-Pio,
    bib/backup
}
\end{spacing}
\end{document}